\documentclass[11pt,a4paper,abstracton]{scrartcl} 

\usepackage{style}

\begin{document}
	
	\title{Recognizing Graph Search Trees\thanks{The work of this paper was done in the framework of a bilateral project between Brandenburg University of Technology and University of Primorska, financed by German Academic Exchange Service and the Slovenian Research Agency (BI-DE/17-19-18).}}
	
	\author[1]{Jesse Beisegel}
	\author[1]{Carolin Denkert}
	\author[1]{Ekkehard K\"ohler}
	\author[2,3]{Matja\v{z} Krnc\thanks{The author gratefully acknowledge the European Commission for funding the InnoRenew CoE project (Grant Agreement \#739574) under the Horizon2020 Widespread-Teaming program and the Republic of Slovenia (Investment funding of the Republic of Slovenia and the European Union of the European Regional Development Fund).}}
	\author[2]{Nevena Piva\v{c}\thanks{Funded in part by the Slovenian Research Agency (research programs P1-0285 and J1-9110 and Young Researchers Grant).}}
	\author[1]{Robert Scheffler}
	\author[1]{Martin Strehler}
	
	\affil[1]{\normalsize Brandenburg University of Technology, Cottbus, Germany}
	\affil[2]{\normalsize University of Primorska, Koper, Slovenia}
  	\affil[3]{\normalsize Faculty of Information Studies, Novo mesto, Slovenia}

	\date{}

  	\maketitle

 	\begin{abstract}
		Graph searches and the corresponding search trees can exhibit important structural properties and are used in various graph algorithms. The problem of deciding whether a given spanning tree of a graph is a search tree of a particular search on this graph was introduced by Hagerup and Nowak in 1985, and independently by Korach and Ostfeld in 1989 where the authors showed that this problem is efficiently solvable for DFS trees. A linear time algorithm for BFS trees was obtained by Manber in 1990. In this paper we prove that the search tree problem is also in \P~for LDFS, in contrast to LBFS, MCS, and MNS, where we show \NP-completeness. We complement our results by providing linear time algorithms for these searches on split graphs. 
	\end{abstract} 	
	
\section{Introduction}

\paragraph*{Motivation.}

Graph searches like \emph{Breadth First Search} (BFS) and \emph{Depth First Search} (DFS) are, in the most general sense, mechanisms for systematically visiting all vertices of a graph. Considered as some of the most basic algorithms in computer science, graph searches are taught in many undergraduate courses around the world and represent an elementary component of several graph algorithms, such as finding connected components, testing for bipartiteness, computing shortest paths with respect to the number of edges, or the Edmonds-Karp algorithm for computing the maximum flow in a network~\cite{edmonds1972theoretical}. Similarly, DFS is the basis for algorithms for finding biconnected components in undirected graphs~\cite{hopcroft1973algorithm}, strongly connected components in directed graphs~\cite{tarjan1972depth}, topological orderings of directed acyclic graphs~\cite{tarjan1976edge}, planarity testing~\cite{hopcroft1974efficient}, or solving mazes~\cite{even2011book}.

We focus on connected searches, that is, a graph search or graph traversal that starts at a vertex and explores the graph by visiting a vertex in the neighborhood of the already visited vertices. If no further restriction is given, we call such a search a \emph{generic search}. The search paradigms of BFS and DFS can be simply characterized by using a queue or a stack as the data structure for the  unvisited vertices in the current neighborhood. However, there are more sophisticated searches like \emph{Lexicographic Breadth First Search} (LBFS)~\cite{rose1976}  and \emph{Lexicographic Depth First Search} (LDFS)~\cite{corneil2008unified}. In this article, we also consider \emph{Maximum Cardinality Search} (MCS)~\cite{tarjan1984simple} and \emph{Maximum Neighborhood Search} (MNS)~\cite{corneil2008unified}.

Usually, the outcome of a graph search is a \emph{search order}, i.e., a sequence of the vertices in the order they are visited. There are many known results and algorithms that are based on graph search orders. For instance, a perfect elimination order of a chordal graph can be found by reversing an LBFS order on that graph~\cite{rose1976}. Apart from a linear recognition algorithm for chordal graphs, LBFS also yields a greedy coloring algorithm for finding a minimum coloring for this graph class~\cite{golumbic2004book}. Furthermore, it is possible to generate characterizing vertex orderings for AT-free graphs using BFS~\cite{beisegel2018characterising}.

A structure that is closely related to a graph search is the corresponding search tree. Such trees can be of particular interest, as for instance the tree obtained by a BFS contains the shortest paths from the root $r$ to all other vertices in the graph. The trees generated by DFS can be used for fast planarity testing of graphs~\cite{hopcroft1974efficient}. Moreover, if a cocomparability graph has hamiltonian path, then such a path can be found by a combination of various graph searches~\cite{corneil2013}. First, one can use at most $n$ LBFS runs, where $n$ is the number of vertices, to find a cocomparability ordering~\cite{dusart2017new}. Afterwards, the last visited vertex of an LDFS on this cocomparability ordering is the first vertex of a hamiltonian path. Finally, the search tree of a right most neighbor search on the LDFS ordering is a hamiltonian path.

So far, there is no satisfactory answer as to why graph searching works so well. An interesting example are multi-sweep algorithms, such as finding dominating pairs in connected asteroidal triple-free graphs~\cite{corneil99}. One can prove that these algorithms are correct. However, it is not clear why multiple runs of a simple algorithm could give such a strong insight into graph structure. Indeed, there seem to be some hidden structural properties of graph searches, which are waiting for discovery and algorithmic exploitation.

As a step in this direction, we study the problem of whether a given tree can be a search tree of a particular search. For BFS-like searches, one usually connects each vertex $v\in V$ to its neighbor which appeared first in the BFS order. Contrary, for DFS-like searches, one connects each vertex $v\in V$ to the last neighbor visited before $v$. However, there is no such obvious definition of a tree for MCS or MNS. Therefore, we define \cf- and \cl-trees: Given an ordering, in an \cf-tree each vertex $v$ is connected to its neighbor which appeared first in the ordering before $v$, whereas in an \cl-tree each vertex is connected to its neighbor which appeared last before $v$. A proper definition will be given in Section~\ref{sec:search:tree}. This motivates the following decision problem:

\begin{problem}{\cf-Tree (\cl-Tree) Recognition Problem}
	Instance: & A connected graph $ G=(V,E) $ and a spanning tree $ T $. \\ 
	Task: & Decide whether there is a graph search of the given type such that $ T $ is \\
		  & its \cf-tree (\cl-tree) of $ G $.
\end{problem}

\paragraph*{Related work.}

Already in 1972, Tarjan~\cite{tarjan1972depth} gave a complete characterization of DFS trees as so-called palm trees. However, no algorithm that determines if a given spanning tree of a graph $G$ is a DFS tree of $G$ was specified in that work. Using the concept of palm trees, Hopcroft and Tarjan developed a linear time algorithm for testing planarity of a graph~\cite{hopcroft1974efficient}. Exploiting properties of DFS and BFS trees, the problem of checking whether a given spanning tree of $G$ can be obtained by a DFS on $G$ was formulated by Hagerup and Novak ~\cite{hagerup1985}. A few years later, Korach and Ostfeld gave a linear time algorithm for the proposed problem of recognition of DFS-trees \cite{korach1989}. A similar result for the recognition of BFS-trees was given by Manber in 1990~\cite{manber1990}.

A problem that is closely related to the search tree recognition problem is the so-called end-vertex problem, i.e., the problem of determining whether a given vertex $v$ in a graph $G$ can be visited last by some graph search method. As a result of numerous new applications in algorithms, the end-vertex problem has received some attention in recent literature. In particular, the end-vertex of an LBFS on a chordal graph is always simplicial~\cite{rose1976}. Furthermore, in a cocomparability graph, the end-vertex of an LBFS is a source/sink in some transitive orientation of its complement~\cite{habib2000}. End-vertices are of particular interest for multi-sweep algorithms, as every consecutive search starts at the end vertex of the previous search. Here, LBFS provides a linear time algorithm for finding dominating pairs in connected asteroidal triple-free graphs, where a dominating pair is a pair of vertices such that every path connecting them is a dominating set in the graph ~\cite{corneil99}. The first vertex $x$ is simply the end-vertex of an arbitrary LBFS and the second vertex $y$ is the end-vertex of an LBFS starting in $x$. Moreover, one can use five LBFS executions followed by a modified LBFS to recognize interval graphs~\cite{corneil2009lbfs}. Crescenzi et al.~\cite{crescenzi2013} have shown that the diameter of huge real world graphs can usually be found with only a few BFS executions.

Surprisingly, the problem of deciding whether a vertex can be an end-vertex of a graph search is hard. In 2010, Corneil, K\"ohler, and Lanlignel~\cite{corneil2010end} have shown that it is \NP-hard to decide whether a vertex can be the end vertex of an LBFS. Later, Charbit, Habib, and Mamcarz generalized this result to BFS, DFS, and LDFS. Furthermore, they extended these results to several graph classes. Recently, Beisegel et~al.~\cite{beisegel2018end-vertex} proved \NP-hardness results for MCS and MNS, and they also provided linear time algorithms for this problem on split graphs and unit interval graphs. 

\paragraph*{Our contribution.}

Although research initially began with the recognition of search trees, the results on the end-vertex problem are currently more extensive. In the light of the new results on the end-vertex problem, we fill in the gaps in the analysis of the complexity of the search tree recognition problem. In this paper, we extend the tree recognition problem to LBFS, LDFS, MCS, and MNS for \cf- or \cl-trees, respectively, by showing \NP-hardness results for most of these searches on general graphs, a polynomial time recognition algorithm for \cl-trees of LDFS on general graphs, and linear time algorithms for the $\cf$-tree and the $\cl$-tree problem on split graphs for various searches. Table~\ref{tab:results} summarizes the known and some of the new results.


\begin{table}[ht!]
	\centering\small
	\begin{tabular}{l c c c c c c c c c c c c}
		\toprule
		Tree results         &     \cf-BFS                 &  \cf-LBFS              &  \cl-DFS & \cl-LDFS         & \cf-MCS             & \cf-MNS                              \\ \midrule
		All Graphs           &     L \cite{manber1990} &  \bfseries{NPC}    &  L \cite{hagerup1985,
			korach1989}  & \bfseries{P}  & \bfseries{NPC}  & \bfseries{NPC} \\ 
		Weakly Chordal       &     L                   &  \bfseries{NPC} &  L                    & \bfseries{P}  & \bfseries{NPC}               & \bfseries{NPC}              \\
		Chordal              &     L                   &  ?   & L                    & \bfseries{P}  & ?               & ?              \\
		Split                &     L                   &  \bfseries{L}   &  L                    & \bfseries{P}  & \bfseries{L}    & \bfseries{L}   \\ \bottomrule \addlinespace
	\end{tabular}
	\caption{Complexity of the tree recognition problem. Our results are denoted by bold letters and L denotes linear time algorithms. \cf\ and \cl\ indicate whether the search is considered with an \cf- or an \cl-tree. }\label{tab:results}
\end{table}

This paper is organized as follows: First, we provide the necessary definitions in Section~\ref{sec:prelim}. An overview of the considered graph searches is given afterwards. In Section~\ref{sec:ldfs} we present a polynomial time algorithm for the \cl-tree problem of LDFS. Sections~\ref{sec:lbfs} and~\ref{sec:mnsmcs} are dedicated to the \NP-completeness of the \cf-tree problem for LBFS, MCS and MNS. In Section~\ref{sec:split} we give the linear time algorithms for split graphs. We conclude the paper with some related open problems.

\section{Preliminaries}\label{sec:prelim}

\subsection{General Notation}

All graphs considered in this paper are finite, undirected, simple and connected. Given a graph $G=(V,E)$, we denote by $n$ and $m$ the number of vertices and edges in $G$, respectively. For a vertex $v\in V$, we denote by $N(v)$ the \emph{neighborhood} of $v$, i.e., the set $N(v)=\{u\in V\mid uv\in E\}$, where an edge between $u$ and $v$ in $G$ is denoted by $uv$. The \emph{closed neighborhood} of $v$ is the set $N[v]=N(v)\cup \{v\}$. A \emph{clique} in a graph $G$ is a set of pairwise adjacent vertices and an \emph{independent set} in $G$ is a set of pairwise nonadjacent vertices. If the neighborhood of a vertex $v$ in $G$ is a clique, then $v$ is said to be a \emph{simplicial vertex}.  The \emph{complement} of the graph $G$ is the simple graph $\overline{G}$ having the same set of vertices as $G$ where for $x,y\in V$, we have that $xy$ is an edge of $\overline{G}$ if and only if it is not an edge in $G$. For a graph $G = (V,E)$ and an edge $e=uv$, where $u$ and $v$ are nonadjacent vertices in $G$, we define $G + e$  to be a graph with vertex set $V$ and edge set $E \cup \{e\}$.

Given a subset $S$ of vertices in $G$, we denote by $G[S]$ the \emph{subgraph of $G$ induced by $S$}, where $V(G[S])=S$ and $E(G[S])=\{xy\in E(G) \mid x\in S, y\in S\}$. By $G-S$ we denote the graph induced by $V(G)\setminus S$. If $S$ contains just one element $v$, we will simply write $G-v$ to denote the graph induced by $V(G)\setminus \{v\}$.


A graph $G$ that contains no induced cycle of length larger than $3$ is called \emph{chordal}. If neither $G$ nor its complement contains an induced cycle of length $5$ or more, then $G$ is said to be \emph{weakly chordal}. A \emph{two-pair} in a graph is a pair of non-adjacent vertices such that every induced path between the two vertices has exactly two edges. We use the following fact about weakly chordal graphs:

\begin{lemma}\cite{spinrad1995algorithms}\label{intro:lemma1}
	Let $ G =(V,E)$ be a graph with a two pair $ \{x,y\} $. Then $ G $ is weakly chordal if and only if $ G+xy $ is weakly chordal.  
\end{lemma}

Similarly, the deletion of some particular vertices does not destroy the property of being weakly chordal.

\begin{lemma}\label{intro:lemma2}
	Let $ G =(V,E) $ be a graph and $ v \in V $ such that $ v $ is simplicial or adjacent to at least $ n-2 $ vertices of $ V $. Then $ G $ is weakly chordal if and only if $ G - v $ is weakly chordal.
\end{lemma}
\begin{proof}
	If $ v $ is simplicial then it cannot be part of an induced cycle of $ G $ of size $ \geq 4 $. Suppose that $ v $ is part of an induced cycle of size $ \geq 5 $ in $ \overline{G} $. Then there is an edge $ uw $ in this cycle, such that $ vu, vw \notin E(\overline{G}) $, a contradiction to $ v $ being simplicial.
	
	Suppose that $ v $ has at least $ n-2 $ neighbors in $ G $. Then $ v $ has only one neighbor in $ \overline{G} $ and, thus, cannot be part of an induced cycle. Suppose $ v $ is part of an induced cycle of size $ \geq 5 $ in $ G $. Then $ v $ must be non-adjacent to at least two vertices, a contradiction.
\end{proof}

A \emph{split graph} $G$ is a graph whose vertex set can be divided into sets $C$ and $I$ such that $C$ is a clique in $G$ and $I$ is an independent set in $G$. It is easy to see, that every split graph is chordal, whereas every chordal graph is also weakly chordal.

An ordering of vertices in $G$ is a bijection $\sigma: V(G) \rightarrow \{1,2,\dots,n\}$. For an arbitrary ordering $\sigma$ of vertices in $G$, we denote by $\sigma(v)$ the position of vertex $v\in V(G)$. Given two vertices $u$ and $v$ in $G$ we say that $u$ is \emph{to the left} (resp. \emph{to the right}) of $v$ if $\sigma(u)<\sigma(v)$ (resp. $\sigma(u)>\sigma(v)$) and we denote this by $u \prec_{\sigma}v$ (resp.  $u \succ_{\sigma}v$). 

A \emph{tree} is an acyclic connected graph. A \emph{spanning tree} of a graph $G$ is an acyclic connected subgraph of $G$ which contains all vertices of $G$. A tree together with a distinguished \emph{root vertex} $r$ is said to be \emph{rooted}. In such a rooted tree a vertex $v$ is an \emph{ancestor} of vertex $w$ if $v$ is an element of the unique path from $w$ to the root $r$. In particular, if $ v $ is adjacent to $ w $, it is called the \emph{parent} of $ w $. Furthermore, a vertex $ w $ is called the \emph{descendant (child)} of $ v $ if $ v $ is the ancestor (parent) of $ w $. A tree is a caterpillar tree, if and only if it admits a dominating path $P$, i.e., every vertex is either in $P$ or adjacent to a vertex in $P$.

\subsection{Graph Searches}

In 1976 Rose, Tarjan and Lueker defined a linear time algorithm (Lex-P) which computes a perfect elimination ordering of a graph if any exists. This algorithm is known as \emph{Lexicographic Breadth First Search} (LBFS) and yields a linear time recognition algorithm for chordal graphs~\cite{rose1976}. LBFS exhibits many interesting structural properties and has been used as a subroutine in many other recognition and optimization algorithms.

\begin{algorithm2e}[H]
	\KwIn{Connected graph $G=(V,E)$ and a distinguished vertex $ s \in V $}
	\KwOut{A vertex ordering $ \sigma $}
	\Begin{
		$ label(s) \leftarrow n $\;
		
		\lFor{each vertex $v \in V-{s}$}{$ label(v) \leftarrow \emptyset $}
		
		\For{$ i \leftarrow 1 $ to $ n $}{pick an unnumbered vertex $ v $ with lexicographically largest label\;\label{slice}
			$ \sigma(i) \leftarrow v$\;
			\For{each unnumbered vertex $ w \in N(v) $}{append $ (n-i) $ to $ label(w) $\;}}
		
	}\caption{Lexicographic Breadth First Search}\index{LBFS}
	\label{lbfs}
\end{algorithm2e}

\emph{Maximum Cardinality Search} (MCS) was introduced in 1984 by Tarjan and Yannakakis~\cite{tarjan1984simple} as a simple alternative to LBFS for recognizing chordal graphs. They noticed that, instead of remembering the \emph{order} in which previous neighbors of a vertex had appeared, it sufficed to just store the \emph{number} of previously visited neighbors for each vertex. This observation resulted in an algorithm which has a linear running time and an easy implementation.

\begin{algorithm2e}
	\KwIn{Connected graph $G=(V,E)$ and a distinguished vertex $ s \in V $}
	\KwOut{A vertex ordering $ \sigma $}
	\Begin{		
		\For{$ i \leftarrow 1 $ to $ n $}{pick an unnumbered vertex $ v $ with the most numbered neighbors\;
			$ \sigma(i) \leftarrow v$\;}
		
	}\caption{Maximum Cardinality Search}
	\label{mcs}
\end{algorithm2e}

In~\cite{corneil2008unified}, Corneil and Krueger defined \emph{Lexicographic Depth First Search} as a lexicographic analogue to DFS. Since then, it has been used for many applications, most notably to solve the minimum path cover problem on cocomparability graphs~\cite{corneil2013}.

\begin{algorithm2e}
	\KwIn{Connected graph $G=(V,E)$ and a distinguished vertex $ s \in V $}
	\KwOut{A vertex ordering $ \sigma $}
	\Begin{
		$ label(s) \leftarrow 0 $\;
		
		\lFor{each vertex $v \in V-{s}$}{$ label(v) \leftarrow \emptyset $}
		
		\For{$ i \leftarrow 1 $ to $ n $}{pick an unnumbered vertex $ v $ with lexicographically largest label\;
			$ \sigma(i) \leftarrow v$\;
			\lForEach{unnumbered vertex $ w \in N(v) $}{prepend $ i $ to $ label(w) $}}
		
	}\caption{Lexicographic Depth First Search}\index{LDFS}
	\label{ldfs}
\end{algorithm2e}

\emph{Maximum Neighborhood Search (MNS)} was introduced by Corneil and Krueger~\cite{corneil2008unified} in 2008 as a generalization of LBFS, LDFS and MCS. Instead of using strings (like LBFS and LDFS) or integers (like MCS) the algorithm uses sets of integers as labels and the maximal labels are those sets which are inclusion maximal. Unlike the labels of LBFS, LDFS and MCS, the labels of MNS are not totally ordered and there can be many different maximal labels. Corneil and Krueger showed that every search ordering of LBFS, LDFS and MCS is also an MNS ordering. This result was generalized in 2009 by Berry et al.~\cite{Berry2009mls} who showed that the set of MNS orderings is equal to the set of orderings of \emph{Maximum Label Search}.

\begin{algorithm2e}
	\KwIn{Connected graph $G=(V,E)$ and a distinguished vertex $ s \in V $}
	\KwOut{A vertex ordering $ \sigma $}
	\Begin{
		assign the label $ \emptyset $ to all vertices\;
		$ label(s) \leftarrow \{n+1\} $\;
		\For{$ i \leftarrow 1 $ to $ n $}{pick an unnumbered vertex $ v $ with maximal label under set inclusion\;
			$ \sigma(i) \leftarrow v$\;
			\lForEach{unnumbered vertex $ w $ adjacent to $ v $}{add $ i $ to $ label(w) $}	
		}	
	}\caption{Maximum Neighborhood Search}
	\label{algo:mns}
\end{algorithm2e}

The relationship between the various searches can be found in Figure~\ref{fig:searches}. For instance, observe that any LBFS, LDFS or MCS is also an MNS. However, the opposite does not hold.

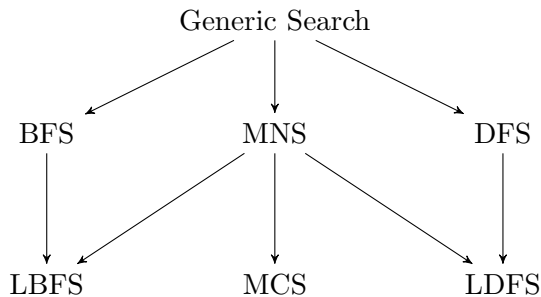
\begin{figure}
	\centering
	\begin{tikzpicture}
	\node (gs) at (3,5.5) {Generic Search};%
	\node (bfs) at (0,4) {BFS};%
	\node (dfs) at (6,4) {DFS};%
	\node (mns) at (3,4) {MNS};%
	\node (mcs) at (3,2) {MCS};%
	\node (lbfs) at (0,2) {LBFS};%
	\node (ldfs) at (6,2) {LDFS};%
	\draw[-stealth'] (gs)--(bfs);%
	\draw[-stealth'] (gs)--(mns);%
	\draw[-stealth'] (gs)--(dfs);%
	\draw[-stealth'] (bfs)--(lbfs);%
	\draw[-stealth'] (mns)--(mcs);%
	\draw[-stealth'] (dfs)--(ldfs);%
	\draw[-stealth'] (mns)--(lbfs);%
	\draw[-stealth'] (mns)--(ldfs);%
	\end{tikzpicture}\caption{This figure represents the relationships between graph searches. The arrows represent proper inclusions. Thus, for example the arrow between BFS and LBFS implies that every LBFS is also a BFS. Searches on the same level are incomparable.}\label{fig:searches}
\end{figure}

\subsection{The Search Tree Recognition Problem}\label{sec:search:tree}

The definition of the term \emph{search tree} varies between different paradigms. However, typically, it consists of the vertices of the graph and, given the search order $ (v_1, \ldots , v_n) $, for each vertex $ v_i $ exactly one edge to a $ v_j \in N(v_i) $ with $ j < i $. By specifying to which of the previously visited neighbors a new vertex is adjacent in the tree, we can define different types of graph search trees. For example, in a BFS a vertex is typically adjacent to the leftmost neighbor in the search order, while in DFS a vertex $ v $ is adjacent to the rightmost neighbor to the left of $ v $. This motivates the following definition.

\begin{definition}
	Given a search discovery order $ \sigma:=(v_1, \ldots , v_n) $ of a given search on a connected graph $ G=(V,E) $, we define the \emph{first-in tree} (or \cf-tree) to be the tree consisting of the vertex set $ V $ and an edge from each vertex to its leftmost neighbor in~$ \sigma $.
	
	The \emph{last-in tree} (or \cl-tree) is the tree consisting of the vertex set $ V $ and an edge from each vertex $ v_i $ to its rightmost neighbor $ v_j $ in $ \sigma $ with $ j < i $.
\end{definition}

As explained above, if $\sigma$ and $T$ are the output of a classical BFS, then $T$ is an \cf-tree with respect to $\sigma$, while for a classical DFS the tree $T$ is an \cl-tree with respect to~$\sigma$. Given this definition, we can state the following decision problem.

\begin{problem}{\cf-Tree (\cl-Tree) Recognition Problem}
	Instance: & A connected graph $ G=(V,E) $ and a spanning tree $ T $. \\ 
	Task: & Decide whether there is a graph search of the given type such that $ T $ is \\
	& its \cf-tree (\cl-tree) of $ G $.
\end{problem}

When comparing the different searches, one can see that graph search trees behave very similarly to the searches themselves, in the sense that, for example, an LBFS tree is also a BFS tree, but not vice versa. Some examples of graph search trees illustrating these relationships can be found in Figure~\ref{prel:fig1}.

\begin{figure}
	\centering
	\begin{tikzpicture}[vertex/.style={inner sep=2pt,draw,circle}]		
	\begin{scope}
	\node[vertex] (1) at (0,0) {};
	\node[vertex] (2) at (0.75,0) {};
	\node[vertex] (3) at (1.5,0) {};
	\node[vertex] (4) at (0.75,1) {};
	\node[vertex] (5) at (0.375,-1) {};
	\node[vertex] (6) at (1.125,-1) {};
	\node[] (a) at (0,1) {a)};
	
	\draw[] (1)--(2)--(3)--(4)--(2)--(6)--(3)--(2)--(5)--(1)--(4);
	\draw[line width =2] (5)--(1)--(4)--(2)--(4)--(3)--(6);
	\end{scope}
	\begin{scope}[xshift=3cm]
	\node[vertex] (1) at (0,0) {};
	\node[vertex] (2) at (0.75,0) {};
	\node[vertex] (3) at (1.5,0) {};
	\node[vertex] (4) at (0.75,1) {};
	\node[vertex] (5) at (0,-1) {};
	\node[vertex] (6) at (0.75,-1) {};
	\node[vertex] (7) at (1.5,-1) {};
	
	\node[] (b) at (0,1) {b)};
	
	\node[vertex] (8) at (0.375,-2) {};
	\node[vertex] (9) at (1.125,-2) {};
	
	\draw[line width=2] (8)--(5)--(1)--(4)--(2)--(6)--(2)--(4)--(3)--(7)--(9);
	\draw[] (1)--(2)--(3);
	\draw[] (8)--(6)--(9);
	\end{scope}
	\begin{scope}[xshift=6cm]
	\node[vertex] (1) at (0,0) {};
	\node[vertex] (2) at (0.75,0) {};
	\node[vertex] (3) at (1.5,0) {};
	\node[vertex] (4) at (0.75,1) {};
	\node[vertex] (5) at (0,-1) {};
	\node[vertex] (6) at (1.5,-1) {};
	\node[vertex] (7) at (0.75,-2) {};
	
	\node[] (c) at (0,1) {c)};
	
	\draw[line width =2] (7)--(6)--(3)--(4)--(2)--(4)--(1)--(5);
	\draw[] (7)--(5)--(2)--(1) --(2)--(3);
	\draw[] (1) to [bend right=45] (3);
	\end{scope}
	\begin{scope}[xshift=9cm]
	\node[vertex] (1) at (0,0) {};
	\node[vertex] (2) at (1.5,0) {};
	\node[vertex] (3) at (0.75,0.75) {};
	\node[vertex] (4) at (0.75,1.5) {};		
	\node[vertex] (5) at (0.75,-0.75) {};
	\node[vertex] (6) at (0.75,-1.5) {};
	\node[vertex] (7) at (0.75,-2.25) {};
	
	\node[] (d) at (0,1) {d)};
	
	\draw[line width=2] (4)--(3)--(5)--(6)--(7)--(6)--(2)--(1);
	\draw[] (1)--(3)--(2)--(5)--(1);
	\end{scope}
	\end{tikzpicture}\caption{Four examples of graphs with their search trees denoted by the thick edges. The graph in a) depicts a search tree of BFS that is not an \cf-tree for LBFS or MNS. The graph in b) depicts an \cf-tree of MNS and BFS that is not an \cf-tree for LBFS. The graph in c) shows a search tree that is an \cf-tree of MNS, BFS and LBFS that is not an \cf-tree of MCS. Finally, the graph in d) gives an example of a search tree that is an \cl-tree for DFS, but not for LDFS.}\label{prel:fig1}
\end{figure}
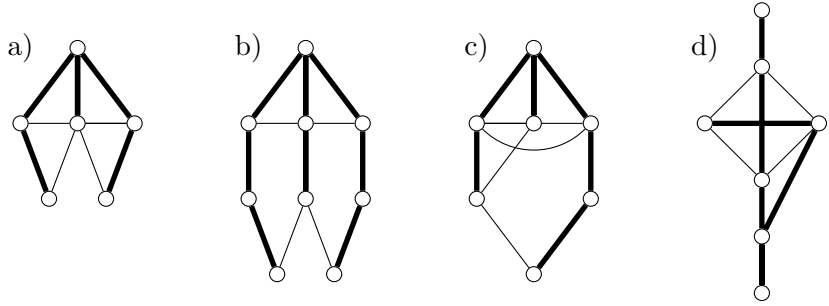

\section{A Polynomial Algorithm for Lexicographic Depth First Search}\label{sec:ldfs}

As Lexicographic Depth First Search is a special case of DFS, the most natural search tree to be considered here is the \cl-tree. We give a polynomial-time algorithm (Algorithm~\ref{ldfs-tree}) which, given a graph $G$ and its spanning tree $T$, decides whether $T$ is an \cl-tree of LDFS on $G$. This is an interesting contrast to the fact that it is \NP-complete to decide whether a given vertex is an end-vertex of LDFS, as shown by Charbit et al.~\cite{charbit2014influence}.

In essence, Algorithm~\ref{ldfs-tree} runs an LDFS and at every step checks whether there is still a possible choice of vertex which does not contradict the search tree.

%
%
%
%
%
%
%

\begin{algorithm2e}
	\KwIn{Graph $G=(V,E)$, spanning tree $T$ of $G$, and a vertex $ r \in V $.}
	\KwOut{$T$ is an \cl-tree of LDFS on $G$ or not.}
	\Begin{			
			$S \leftarrow  \{r\} $\;
			
			\For{each vertex $v \in V-{r}$}{$ \lab(v) \leftarrow \emptyset $\;}
			
			\For{each vertex $v \in N(r)$}{prepend 0 to $ \lab(v) $\;
				$ \pred(v) \leftarrow r $\;
			}
			
			\While{$ S \neq V$}{choose a node $v \in V-S$ with lexicographic largest label, such that $\{\pred(v),v\}\in E(T)$ \;
				
				\If{no such v exists}{\Return{$T$ is not an \cl-tree of LDFS on $ G $\;}}
				
				$ S \leftarrow S \cup \{v\} $\;
				
				\For{$ w \in N(v) \setminus S $}{prepend $i$ to $ \lab(w) $\;
					$ \pred(w) \leftarrow v $\;
				}
			}		
			\Return{$T$ is an \cl-tree of LDFS on $ G $.}
		}	
			
	\caption{Algorithm which decides whether $T$ is an \cl-tree of LDFS on $ G $.}\label{ldfs-tree}	
\end{algorithm2e}

To prove that Algorithm~\ref{ldfs-tree} works correctly, we first state a few lemmas about \cl-trees of DFS.

\begin{lemma}\cite{tarjan1972depth}\label{ldfs:lemma3}
	Let $ G=(V,E) $ be a graph and let $ T $ be an \cl-tree of $ G $ generated by DFS. For each edge $ uv \in E $ it holds that either $ e \in E(T) $ or, without loss of generality, $ u $ is an ancestor of $ v $ in $ T $.
\end{lemma}

\begin{lemma}\cite{korach1989}\label{ldfs:lemma1}
	Let $ G=(V,E) $ be a graph with spanning tree $ T $. Let $ G_i $ be a subgraph of $ G $ with a spanning tree $ T_i $ which is the restriction of $ T $ to $ G_i $. If $ T $ is an \cl-tree of DFS on $ G $, then $ T_i $ is an \cl-tree of DFS on $ G_i $.
\end{lemma}

We can give an analogous result for LDFS, which just considers induced subgraphs of $ G $.

\begin{lemma}\label{ldfs:lemma2}
	Let $ G=(V,E) $ be a graph with spanning tree $ T $. Let $ G_i $ be an induced subgraph of $ G $ with a spanning tree $ T_i $ which is the restriction of $ T $ to $ G_i $. If $ T $ is an \cl-tree of LDFS on $ G $, then $ T_i $ is an \cl-tree of LDFS on $ G_i $. In particular, if $ T $ is rooted in $ r \in V $ and $ r \in V(T_i) $, then $ T_i $ is also rooted in $ r $. 
\end{lemma}
\begin{proof}
	Let $ G_i $ be an induced subgraph of $ G $ and let $ T_i $ be the restriction of $ T $ to $ G_i $. Suppose that $ T $ is an \cl-tree of LDFS on $ G $. We will show that in this case $ T_i $ is an \cl-tree of LDFS on $ G_i $.
	
	Let $ \sigma $ be an LDFS search order of $ G $ that results in the search tree $ T $ and let $ \sigma(1):=r $. We run an LDFS on $ G_i $ by always choosing the vertex with largest label which is leftmost in $ \sigma $ and call the new search order $ \tau $. Suppose that the resulting search tree $ R $ does not coincide with $ T_i $. Let $ v $ be the leftmost vertex in $ \tau $ that does not have the same parent in $ R $ as it does in $ T $. Let $ u $ be the parent of $ v $ in $ R $.
	
	Because $ v $ was chosen to be leftmost in $ \tau $ such that it has a different parent in $ R $ than in $ T $, the unique path $ P $ from $ u $ to $ r $ in $ R $ is identical to that in $ T $. Therefore, we can see that $ u $ must be an ancestor of $ v $ in $ T $, due to Lemma~\ref{ldfs:lemma3}. Let $ w $ be the unique child of $ u $ in $ T $ that is an  ancestor of $ v $; in particular $ v \neq w $ and $ w \prec_{\sigma} v $. As $ v $ is a child of $ u $ in $ R $, we can assume that $ v \prec_{\tau} w $. This implies that at the point where $ v $ was chosen, the label of $ v $ was strictly larger than that of $ w $; this is a contradiction, as all vertices that have labeled $ v $ are on $ P $, due to Lemma~\ref{ldfs:lemma3}. Therefore, it is identical to the label $ v $ had at the point when $ w $ was chosen over $ v $ in $ \sigma $.	
\end{proof}

\begin{theorem}
	The \cl-tree recognition problem for LDFS can be solved in polynomial time.
\end{theorem}
\begin{proof}
	Algorithm~\ref{ldfs-tree} tests for a fixed $ r \in V $ whether $ T $ can be an \cl-tree for LDFS on $ G $ that is rooted in $ r $. Therefore, assuming the Algorithm~\ref{ldfs-tree} works correctly and in polynomial time, it is enough to apply it to all vertices in $ G $ to decide whether $ T $ is, in fact, an \cl-tree of LDFS. As we begin the search in $ r $ we from now on assume that $ T $ is rooted in a fixed vertex $ r $.
	
	First suppose that the algorithm returns ``$T$ is an \cl-tree of LDFS on $ G $''. In this case, the algorithm has successfully executed an LDFS and it remains to show that the resulting search order has $ T $ as its \cl-tree. This, however, is safeguarded by the fact that at every point at which we have added a vertex $ v $ to our search order, the predecessor of $ v $, i.e., its parent in the resulting search tree, is also adjacent to $ v $ in $ T $.
	
	Now assume that the algorithm returns ``$T$ is not an \cl-tree of LDFS on $ G $''. This implies that at some point of the LDFS there is no vertex $ x $ of lexicographically largest label, such that the predecessor of $ x $ is adjacent to $ x $ in $ T $. Let $ v $ be such a vertex of lexicographically largest label, whose predecessor is not its parent in $ T $. As $ v $ is the first such vertex to appear in the search, the tree $ R $ constructed thus far by Algorithm~\ref{ldfs-tree} is a subtree of $ T $.
	
	Assume that $ T $ is, in fact, an \cl-tree of $ G $ generated by LDFS. Let $ u $ be the predecessor assigned to $ v $ by the algorithm. Thus, due to Lemma~\ref{ldfs:lemma3}, $ u $ must be an ancestor of $ v $ in $ T $. Let $ w $ be the unique child of $ u $ in $ T $ that is also an ancestor of $ v $ and let $ P $ be the unique path from $ v $ to $ r $ in $ T $; in particular, $ u,w \in V(P) $. As a result of Lemma~\ref{ldfs:lemma2}, $ P $ is an \cl-tree of LDFS on $ G[V(P)] $ since $ T $ is an \cl-tree of LDFS on $ G $.
	
	However, Algorithm~\ref{ldfs-tree} and Lemma~\ref{ldfs:lemma3} imply that $ P $ cannot be an \cl-tree of LDFS on $ G[V(P)] $: As we start in $ r $ and as $ P $ is a path, we must choose all vertices up to $ u $ in the order of the path. Due to Lemma~\ref{ldfs:lemma3}, the vertices have the same labels as they did when Algorithm~\ref{ldfs-tree} halted. Therefore, $ v $ has a lexicographically larger label than $ w $. As a result, $ P $ and, thus, $ T $ cannot be a \cl-trees of LDFS.
\end{proof}

\section{\NP-Completeness for Lexicographic Breadth First Search}\label{sec:lbfs}

It was shown in \cite{corneil2010end} that the LBFS end-vertex problem is \NP-complete. In the following we show that the same holds for the tree-recognition problem.

\begin{theorem}\label{lbfs:theorem1}
	The \cf-tree-recognition problem of LBFS is \NP-complete on weakly chordal graphs.
\end{theorem}

We prove Theorem~\ref{lbfs:theorem1} by giving a reduction from 3-SAT. Let $\I$ be an instance of 3-SAT. We construct the corresponding graph $G(\I)$ and the spanning tree $ T(\I) $ as follows (for an example see Figure~\ref{lbfs:fig1}): Let $ X=\{x_1, \dots, x_k,\overline{x_1},\ldots,\overline{x_k}\} $ be the set of vertices representing the literals of $ \I $. The edge set $ E(X) $ forms the complement of the matching in which $ x_i $ is matched to $ \overline{x_i} $ for every $ i \in \{1,\ldots, k\} $. For each clause $ C_i $ of $ \I $ we have a triangle consisting of vertices $ a_i $, $ c_i $ and $ t_i $. For every triangle representing a clause $ C_i $, the vertex $ c_i $ is adjacent to each literal of the clause $ C_i $.

In addition, we have vertices $ r $, $ p $, $ q $ and $ u $. Vertex $ r $ is adjacent to every vertex apart from the $ t_i $ and $ u $, while $ u $ is adjacent to all vertices apart from the $ t_i $ and $ r $. Vertex $ p $ has additional edges to each vertex in $ X $ and to $ q $, while $ q $ is also adjacent to all vertices in $ X $ and each of the $ a_i $. Altogether, $ G(\I) $ consists of the vertex set $ V(G(\I)) := X \cup \{r,p,q,u\}\cup C_1 \cup \ldots \cup C_l$, where $ C_i $ represents the vertices of the clause-gadget of $ C_i $ and the edge set is defined as above.

The corresponding spanning tree $ T(\I) $ consists of the edges incident to $ r $, an edge between $ u $ and $ p $ and the edges $ c_it_i $ for all $ i \in \{1, \ldots l\} $; they are denoted as thick lines in Figure~\ref{lbfs:fig1}.

\begin{figure}[H]
	\centering
	\begin{tikzpicture}
	[vertex/.style={inner sep=2pt,draw,circle,fill=white},
	noedge/.style={dashed},
	yscale=0.9
	]
	
	\node[vertex,label={0:$r$}] (r) at (3,6.5) {}
	edge[line width=2] (3.6,6)
	edge[line width=2] (2.4,6)
	edge[line width=2] (2.8,6)
	edge[line width=2] (3.2,6);
	
	\draw[rounded corners=5pt, fill = white]
	(-1,0.4) rectangle (8.5,6);
	
	\begin{scope}[scale=0.75, xshift=0cm, yshift=0cm]
	\draw[rounded corners=5pt, fill=black!10!white]
	(-0.5,-0.5) rectangle (2.5,2.5);

	\node[vertex,label=below:\footnotesize$ t_1 $] (t1) at (1.75,0.25){};
	\node[vertex,label=below:\footnotesize$ a_1 $] (a1) at (0.25,1.75){};
	\node[vertex,label=right:\footnotesize$ c_1 $] (c1) at (1.75,1.75){};
	
	\draw[](t1)--(a1)--(c1);
	\draw[line width=2](t1)--(c1);
	
	\node at (1,-1) { $\overline{x_1}\lor x_2\lor \overline{x_3}$};
	\end{scope}
	
	\begin{scope}[scale=0.75, xshift=4cm, yshift=0cm]
	\draw[rounded corners=5pt, fill=black!10!white]
	(-0.5,-0.5) rectangle (2.5,2.5);

	\node[vertex,label=below:\footnotesize$ t_2 $] (t2) at (1.75,0.25){};
	\node[vertex,label=below:\footnotesize$ a_2 $] (a2) at (0.25,1.75){};
	\node[vertex,label=right:\footnotesize$ c_2 $] (c2) at (1.75,1.75){};
	
	\draw[](t2)--(a2)--(c2);
	\draw[line width=2](t2)--(c2);
	
	\node at (1,-1) { $\overline{x_1} \lor \overline{x_3} \lor \overline{x_4}$};
	\end{scope}
	
	\begin{scope}[scale=0.75, xshift=8cm, yshift=0cm]
	\draw[rounded corners=5pt, fill=black!10!white]
	(-0.5,-0.5) rectangle (2.5,2.5);

	\node[vertex,label=below:\footnotesize$ t_3 $] (t3) at (1.75,0.25){};
	\node[vertex,label=below:\footnotesize$ a_3 $] (a3) at (0.25,1.75){};
	\node[vertex,label=right:\footnotesize$ c_3 $] (c3) at (1.75,1.75){};
	
	\draw[](t3)--(a3)--(c3);
	\draw[line width=2](t3)--(c3);
	
	\node at (1,-1) { $\overline{x_1} \lor x_3 \lor \overline{x_4}$};
	\end{scope}

	\begin{scope}[scale=0.75, xshift=3cm, yshift=5cm]
	\draw[rounded corners=5pt, fill=black!10!white]
	(-1,-0.5) rectangle (6.5,2);
	
	\node[vertex,label={180:$x_1$}] (x1) at (0,1.5) {};
	\node[vertex,label={180:$\overline{x_1}$}] (nx1) at (0,0) {}
	edge[] (c1)
	edge[] (c2)
	edge[] (c3)
	edge[noedge] (x1);
	
	\node[vertex,label={180:$x_2$}] (x2) at (2,1.5) {}
	edge[] (c1);
	\node[vertex,label={180:$\overline{x_2}$}] (nx2) at (2,0) {}
	edge[noedge] (x2);

	\node[vertex,label={180:$x_3$}] (x3) at (4,1.5) {}
	edge[] (c3);
	\node[vertex,label={180:$\overline{x_3}$}] (nx3) at (4,0) {}
	edge[] (c1)
	edge[] (c2)
	edge[noedge] (x3);
	
	\node[vertex,label={180:$x_4$}] (x4) at (6,1.5) {};
	\node[vertex,label={180:$\overline{x_4}$}] (nx4) at (6,0) {}
	edge[] (c2)
	edge[] (c3)
	edge[noedge] (x4);
	\end{scope}
	
	\node[vertex,label={180:$q$}] (t) at (0,3) {}
	edge[] (1.5,4.1)
	edge[] (1.5,3.9)
	edge[] (1.5,3.7)
	edge[] (1.5,3.5)
	edge[] (a1)
	edge[] (a2)
	edge[] (a3);
	
	\node[vertex,label={180:$p$}] (p) at (0,5) {}
	edge[] (1.5,5.1)
	edge[] (1.5,4.9)
	edge[] (1.5,4.7)
	edge[] (1.5,4.5)
	edge[] (t);
	
	\node[vertex,label=above:$ u $] (u) at (-2, 3.5) {}
	edge[] (-1,3.4)
	edge[] (-1,3.2)
	edge[] (-1,3.6)
	edge[] (-1,3.8)
	edge[line width=2] (p);
	
	\end{tikzpicture}
	\caption{The \NP-completeness construction for the tree-recognition problem of LBFS. The depicted graph is $G(\I)$ for $\I = (\overline{x_1} \lor x_2 \lor \overline{x_3}) \land (\overline{x_1} \lor \overline{x_3} \lor \overline{x_4}) \land (\overline{x_1} \lor x_3 \lor \overline{x_4})$. In the box containing the literal vertices, only non-edges are displayed by dashed lines. The connection of a vertex with a box implies, that the vertex is connected to all vertices in this box. Tree edges are depicted by thick edges.}\label{lbfs:fig1}
\end{figure}
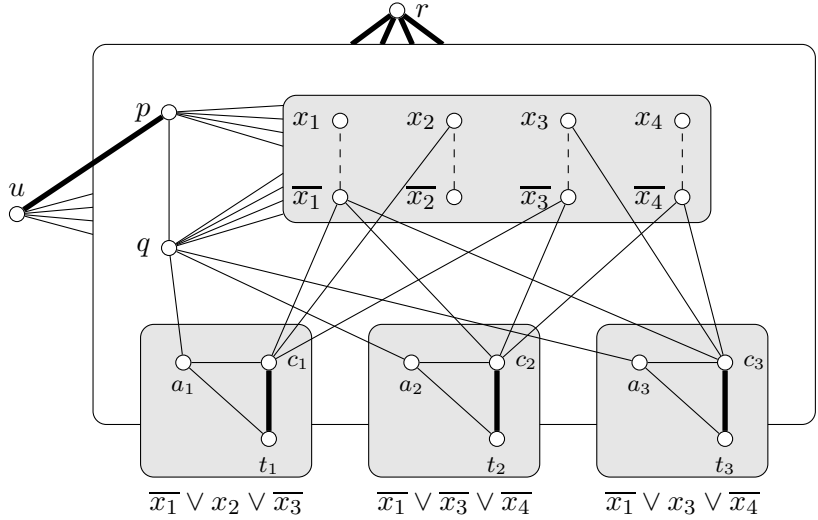

We  proceed to prove Theorem~\ref{lbfs:theorem1} by showing that $ T(\I) $ is an $\cf$-tree of LBFS of $ G(\I) $ if and only if $ \I $ has a satisfying assignment $ \A $.

\begin{lemma}\label{lbfs:lemma1}
	If $ \I $ admits a satisfying assignment $ \A $, then $ T(\I) $ is a possible $\cf$-tree of LBFS on $ G(\I) $.
\end{lemma}
\begin{proof}
	Let $ \A $ be a satisfying assignment of $ \I $. The following valid search order produces $ T(\I) $ as its search tree: We begin in $ r $ and then choose $ p $. Next, we can choose vertices from $ X $ according to the assignment $ \A $ in an arbitrary order, i.e., we choose 
	$ x_i $ or $ \overline{x_i} $ corresponding to whether the variable $x_i$ is set to $1$ or $0$ in $\A$.
	 We are then forced to visit the vertex $ q $, as each remaining vertex of $ X $ is not adjacent to one of the visited vertices of $ X $. After choosing the remaining vertices of $ X $ we proceed to the vertices of the clause gadgets: As a fulfilling assignment sets at least one literal to 1 in each clause, every $ c_i $ has a neighbor that appears earlier in the search order than $ q $ which is the leftmost neighbor of $ a_i $ in the search order. Hence, for each clause gadget $ C_i $ we must choose $ c_i $ before $ a_i $. Therefore, we can choose all vertices $ c_i $ and then all vertices $ a_i $. Finally, we can choose $ u $ and then all the $ t_i $.
	
	It is easy to see that all edges incident to $ r $ belong to the search tree of the constructed order, as well as $ pu $. On the other hand, $ c_it_i $ must be in the search tree for every $ i \in \{1, \ldots, l\} $, as $ c_i $ was always chosen before $ a_i $. Therefore, the search tree of the constructed order coincides with $ T(\I) $.
\end{proof}

We now show the other direction of the proof.

\begin{lemma}\label{lbfs:lemma2}
	If $ \I $ does not admit a satisfying assignment, then $ T(\I) $ cannot be an $\cf$-tree of LBFS on $ G(\I) $.
\end{lemma}
\begin{proof}
	We show that for at least one clause gadget $ C_i $ the vertex $ a_i $ is visited before $ c_i $, thus making $ T(\I) $ an infeasible search tree.
	
	To prove this, we analyze the order in which the vertices of $ X $ are visited in any feasible LBFS search. It is easy to see that any LBFS must begin in $ r $, as $ r $ is the only vertex whose incident edges are all tree edges. Next, we are forced to choose $ p $, as otherwise $ pu $ cannot be a tree edge. If $ q $ is chosen next, then, as a result, $ a_i $ must be visited before $ c_i $ for every $ i \in \{1, \ldots , l\} $ and $ T(\I) $ cannot be the resulting search tree. Therefore, a subset of the vertices of $ X $ must be chosen before the vertex $ q $.
	
	If a vertex $ x_i $ is visited, then $ q $ receives a larger label than $ \overline{x_i} $, as they otherwise share the same set of neighbors among the visited vertices up to that point (and analogously if $ \overline{x_i} $ is visited before $ q $). Thus, $ q $ must be chosen between any literal vertex and its negation. The largest subset of $ X $ that can be visited before $ q $ must, therefore, be an assignment of $ \I $. As $ \I $ is not satisfiable, any such assignment must leave at least one clause unfulfilled. If $ C_i $ is such a clause, then at the point at which $ q $ is chosen, $ c_i $ does not contain any neighbors among the visited literal vertices. As a result, $ a_i $ receives a larger label than $ c_i $ and is visited earlier.
	
	Consequently, in any LBFS there must be a clause $ C_i $ such that $ a_i $ is visited before $ c_i $ and $ c_it_i $ cannot be in the search tree. This shows that $ T(\I) $ cannot be a $\cf$-tree of an LBFS.
\end{proof}

\begin{corollary}
	Let $\I$ be an instance of 3-SAT. Then $ \I $ has a satisfying assignment if and only if $ T(\I) $ is a possible $\cf$-tree of LBFS on $ G(\I) $.
\end{corollary}

To conclude the proof of Theorem~\ref{lbfs:theorem1} it remains to show that $G(\I)$ is weakly chordal for every 3-SAT instance $\I$.

\begin{lemma}\label{lbfs:lemma3}
	For each instance $\I$ of 3-SAT, the graph $G(\I)$ is weakly chordal.
\end{lemma}
\begin{proof}
	We need to show that both $ G(\I) $ and $ \overline{G(\I)} $ do not contain a cycle of length $ \geq 5 $. As all the $ t_i $ are simplicial, we can disregard them, due to Lemma~\ref{intro:lemma2}. In the remaining graph, both $ r $ and $ u $ are adjacent to all vertices apart from each other and can, thus, be deleted, due to Lemma~\ref{intro:lemma2}.
	
	Let $ H' $ be the graph resulting from deleting $ r $, $ u $ and all the $ t_i $; it suffices to show that $ H' $ is weakly chordal. In addition, it is easy to see that every non-edge $ x_i\overline{x_i} $ forms a two-pair in $ H' $, i.e., the longest induced path between these two vertices is of length~2. Using Lemma~\ref{intro:lemma1}, we see that $ H' $ is weakly chordal if and only if $ H'+x_i \overline{x_i}  $ is weakly chordal. Furthermore, if we add the edges $ x_i\overline{x_i} $ for all $ i \in \{1, \ldots ,k\} $ to $ H' $, the vertex $ p $ becomes simplicial. Therefore, it remains to show that the graph $ H $ which is constructed from $ H' $ by adding the edges $ x_i\overline{x_i} $ for all $ i \in \{1, \ldots ,k\} $ and then deleting $ p $ is weakly chordal.
	
	It is sufficient to show that $ \overline{H} $ is weakly chordal. To do this, we again apply Lemma~\ref{intro:lemma2}. We can delete $ q $ from $ \overline{H} $ as it is simplicial. In the remaining graph, all the $ a_i $ are adjacent to all but one vertex and can, thus, also be deleted. The remaining graph is a split graph, as the $ c_i $ form a clique and the literal vertices form an independent set, and, as a result it is weakly chordal.
\end{proof}

\section{\NP-Completeness for Maximum Neighborhood Search and Maximum Cardinality Search}\label{sec:mnsmcs}

As we have done for LBFS, we will show that the \cf-tree problems for MNS and MCS are \NP-complete.

\begin{figure}[ht!]
		\centering
	\begin{tikzpicture}[vertex/.style={inner sep=2pt,draw,circle,fill=white},
	noedge/.style={dashed}
	]
	
	\draw[rounded corners=5pt]
 	(-2.5,4) rectangle (6.5,9);
 	\draw[rounded corners=5pt, fill=black!10!white]
 	(-2,7) rectangle (7,8.75);

	\node[vertex,label={180:$x_1$}] (x1) at (-1, 8.4) {};
	\node[vertex,label={180:$\overline{x}_1$}] (nx1) at (-1, 7.4) {};
	\node[vertex,label={0:$x_2$}] (x2) at (1, 8.4) {};
	\node[vertex,label={0:$\overline{x}_2$}] (nx2) at (1, 7.4) {};
	\node[vertex,label={180:$x_3$}] (x3) at (3, 8.4) {};
	\node[vertex,label={0:$\overline{x}_3$}] (nx3) at (3, 7.4) {};
	\node[vertex,label={0:$x_4$}] (x4) at (5, 8.4) {};
	\node[vertex,label={0:$\overline{x}_4$}] (nx4) at (5, 7.4) {};
	
	\node[vertex,label={[label distance=0.5cm]270:$\overline{x}_1 \lor x_2 \lor \overline{x}_3$}] (c1) at (-1, 5.25) {};
	\node[vertex,label={[label distance=0.5cm]270:${x}_1 \lor \overline{x}_3 \lor {x}_4$}] (c2) at (2,5.25) {};
	\node[vertex,label={[label distance=0.5cm]270:$\overline{x}_1 \lor \overline{x}_3 \lor \overline{x}_4$}] (c3) at (5, 5.25) {};
	
	\node[vertex,label={0:$t$}] (t) at (10,6.25) {};
	\node[vertex,label={0:$b$}] (b) at (9,8) {};
  \node[vertex,label={0:$a$}] (a) at (7.5,6.5) {};
  \node[vertex,label={0:$p$}] (p) at (7,6) {};
	\node[vertex,label={0:$r$}] (r) at (7.5,5.25) {};
	\node[vertex,label={-90:$q$}] (q) at (9,4.5) {};
	
	\draw[line width=2] (r) -- (q);
	\draw[line width=2] (r) to (b);
  \draw[line width=2] (r) to (p);
  \draw[line width=2] (p) to (a);
  \draw[] (b) to (a);
  \draw[] (q) to (a);
	
	\draw[noedge] (x1) -- (nx1);
	\draw[noedge] (x2) -- (nx2);
	\draw[noedge] (x3) -- (nx3);
	\draw[noedge] (x4) -- (nx4);
  
  \draw[noedge] (c1) -- (c2);
  \draw[noedge] (c2) -- (c3);
  \draw[noedge, bend angle=20, bend right] (c1) to (c3);
	
	\draw[] (q) to (b);
	\draw (q) -- (t);

	\draw[noedge] (c1) -- (nx1);
	\draw[noedge] (c1) -- (x2);
	\draw[noedge] (c1) -- (nx3);
	\draw[noedge] (c2) -- (x1);
	\draw[noedge] (c2) -- (nx3);
	\draw[noedge] (c2) -- (x4);
	\draw[noedge] (c3) -- (nx1);
	\draw[noedge] (c3) to (nx3);
	\draw[noedge] (c3) -- (nx4);
	
	
	
  \draw[] (a) -- (6.5,6.55)--(a)--(6.5,6.45)--(a)--(6.5,6.35)--(a)--(6.5,6.65);
  
  \draw[] (p) -- (6.5,6.05)--(p)--(6.5,5.95)--(p)--(6.5,5.85)--(p)--(6.5,6.15);
  
	\draw[] (b) -- (7,8.1)--(b)--(7,7.9)--(b)--(7,7.7)--(b)--(7,8.3);
	
  \draw[line width=2] (r) -- (6.5,5.35) --(r)--(6.5,5.15)--(r)--(6.5,5.55)--(r)--(6.5,4.95);
  
	\draw[] (q) -- (6.5,4.6) --(q)--(6.5,4.8)--(q)--(6.5,4.4)--(q)--(6.5,4.2);
	
	\draw[line width=2] (b) -- (t);
	
	\end{tikzpicture}
 	\caption{The $ \mathcal{NP} $-completeness construction for the tree-recognition problem of MNS. The depicted graph is $G(\mathcal{I})$ for $\mathcal{I} = (\overline{x}_1 \lor x_2 \lor \overline{x}_3) \land ({x}_1 \lor \overline{x}_3 \lor {x}_4) \land (\overline{x}_1 \lor x_2 \lor \overline{x}_3)$. In both boxes only non-edges are displayed by dashed lines. The connection of a vertex with a box means, that the vertex is connected to all vertices in this box. Tree edges are depicted by thick edges.}\label{fig:mns-tree}
\end{figure}
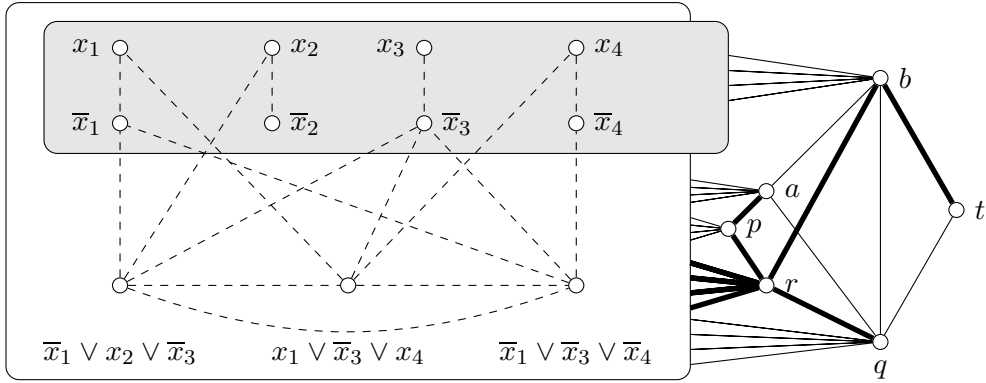

\begin{theorem}\label{theo:mns_npc}
The \cf-tree-recognition problem of MNS and MCS is \NP-complete on weakly chordal graphs.
\end{theorem}

For the proof we construct a polynomial reduction from 3-SAT. Let $\cal{I}$ be an instance of 3-SAT. We construct the corresponding graph $G(\mathcal{I})$ as follows (see Figure~\ref{fig:mns-tree} for an example): Let $ X=\{x_1, \dots, x_k,\overline{x}_1,\ldots,\overline{x}_k\} $ be the set of vertices representing the literals of $ \mathcal{I} $. The edge-set $ E(X) $ forms the complement of the matching in which $x_i$ is matched to $ \overline{x}_i $ for every $ i \in \{1,\ldots, k\} $. Let $ C = \{c_1, \ldots ,c_l\} $ be the set of vertices representing the clauses of $ \mathcal{I} $. The set $ C $ is independent in $G(\mathcal{I})$ and every $ c_i $ is adjacent to each vertex of $ X $, except those representing the literals of the clause associated with $ c_i $ for every $ i \in \{1, \ldots ,l \} $. Additionally, we add the vertices $r$, $p$, $q$, $a$, $b$ and $t$. The vertices $r$, $p$, $q$ and $a$ are adjacent to all literal vertices and all clause vertices and $b$ is adjacent to all literal vertices. Finally, we add the edges $ab$, $ap$, $aq$, $bq$, $br$, $bt$, $pr$, $ qr $ and $qt$. 

The spanning tree $T(\I)$ of $G(\I)$ consists of all edges incident to $r$ and the edges $pa$ and $bt$.

\begin{lemma}\label{lemma:mns_bclause}
If MNS or MCS generates the $\cf$-tree $T(\I)$ on $G(\I)$, it chooses $b$ before every clause vertex $c_i$.
\end{lemma}

\begin{proof}
If we take the vertex $q$ before $b$, we will insert the edge $qt$ to the search tree, which is not an element of $T(\I)$. Thus, this is not allowed in a search that generates the $\cf$-tree $T(\I)$. The neighborhood of $b$ is properly contained in the neighborhood of $q$. Furthermore, $q$ is adjacent to each clause vertex, while $b$ is adjacent to none of them. Hence, if vertex $c_i$ is taken before $b$, then the label of $q$ will always be greater than the label of $b$ in both MNS and MCS and both searches will take $q$ before $b$.
\end{proof}

\begin{lemma}\label{lemma:mns_tree_assignment}
Let $\sigma$ be an MNS ordering of $G(\I)$ that generates the \cf-tree $T(\I)$. Then $\sigma(1) = r$, $\sigma(2) = p$ and $\sigma(i)$ for $3 \leq i \leq k + 2$ forms an arbitrary assignment of the variables (not necessarily satisfying).
\end{lemma}

\begin{proof}
	Any MNS resulting in the search tree $T(\I)$ must start in $r$, since every other vertex is incident to an edge in $G(\I)$ which is not an element of $T(\I)$. Since $a$ is adjacent to every neighbor of $r$ in $G(\I)$ but only to $p$ in $T(\I)$, the search has to choose $p$ as the next vertex. Now the literal vertices and the clause vertices have the unique maximal label, since they were labeled both by $r$ and $p$ and every other vertex was labeled by at most one of these two vertices. Because of Lemma~\ref{lemma:mns_bclause} we cannot take a clause vertex. Thus, we have to take a literal vertex. With the same argumentation it follows that we have to take a whole assignment, since the literal vertices of variables whose two literal vertices have not yet been chosen always have the unique maximal label.  
\end{proof}

\begin{lemma}\label{mns:lemma1}
If $\I$ has a satisfying assignment $\A$, then $T(\I)$ is an \cf-tree of MCS on $G(\I)$ and, therefore, also an \cf-tree of MNS.
\end{lemma}

\begin{proof}
	In the following we give a search order which results in the desired search tree $ T(\I) $. We start with $r$ and then we take $p$. By doing this, we insert every edge of $T(\I)$ apart from $bt$ to the search tree. Next, we take the literal vertices which correspond to the assignment $\A$ in an arbitrary order. As a result, the labels of all literal vertices and of the vertices $a$, $b$ and $q$ are equal to $k + 1$. Since $\A$ is satisfying, each clause vertex was not labeled by at least one of the chosen literal vertices. Hence, it has a label $\leq k +1$ and we can take $b$ as the next vertex and insert the last missing edge of $T(\I)$. The remaining vertices can be chosen in any possible order, as they do not influence the search tree. 
\end{proof}

\begin{lemma}\label{mns:lemma2}
If $\I$ does not have a satisfying assignment, then $T(\I)$ is not an MNS \cf-tree of $G(\I)$ and, therefore, also not an MCS \cf-tree.
\end{lemma}

\begin{proof}
	Assume that $T(\I)$ is an MNS \cf-tree of $G(\I)$. By Lemma~\ref{lemma:mns_tree_assignment} we have to start with $r$, then $p$ and, next, the literal vertices that correspond to an arbitrary assignment. Since this assignment cannot be satisfying, there is at least one clause vertex which was labeled by every vertex chosen up till now. In the label of every non-clause vertex at least one chosen vertex is missing. Thus, we have to visit a clause vertex next. This contradicts Lemma~\ref{lemma:mns_bclause}.
\end{proof}

\begin{lemma}\label{mns:lemma3}
	For every instance $\I$ of 3-SAT the graph $G(\I)$ is weakly chordal.
\end{lemma}

\begin{proof}
	To begin with, we will use Lemma~\ref{intro:lemma2} to delete some vertices which cannot be part of a cycle of length $\geq 5$ in $G(\I)$ or its complement. We can delete $t$, since it is simplicial. Now the vertices $r$ and $a$ are adjacent to every other vertex and, therefore, we can delete these as well. In the resulting graph we can use the same argumentation to delete $q$ and $p$. The remaining graph only contains the literal vertices, the clause vertices and $b$. Since $x_i$ and $\overline{x}_i$ form a two-pair for every $1 \leq i \leq k$, we can add the edges $x_i\overline{x}_i$, due to Lemma~\ref{intro:lemma1}. The resulting graph is a split graph, where $X \cup \{b\}$ forms the clique and $C$ forms the independent set. Thus, it is weakly chordal.
\end{proof}

Theorem~\ref{theo:mns_npc} follows immediately from Lemma~\ref{mns:lemma1}, Lemma~\ref{mns:lemma2} and Lemma~\ref{mns:lemma3}.


\section{Linear Time Algorithms for Split Graphs}\label{sec:split}

Surprisingly, for split graphs the set of \cf-trees is the same for the searches BFS, MNS, MCS, and LDFS, even though this does not hold for the respective search orders. We exploit this special structure to derive a linear time algorithm for split graphs. Note that LDFS is considered together with an \cf-tree.

\begin{theorem}
	A tree $T$ is an \cf-tree of BFS on a split graph $G$ if and only if it is an \cf-tree of MNS (MCS, LBFS, LDFS).
\end{theorem}

\begin{proof}
	Let $ G=(V,E) $ be a split graph and let $ T $ be an \cf-tree for BFS on $G$, generated by the order $\tau$. Let $I=\{i_1,\dots,i_{\ell}\}$ be the independent set and $C=\{c_1,\dots,c_k\}$ be the clique of $ G $. We  show that there is an MNS ordering $ \sigma $ that generates a search tree that coincides with $ T $.
	
	Suppose $\tau$ starts with a clique vertex, without loss of generality~$c_1$, that is, $c_1$ is the root of the search tree. Then, all other clique vertices $c_2$ to $c_k$ are in the first layer of the \cf-tree, and additionally, all independent set vertices which are adjacent to $c_1$ are in the first layer as well. Without loss of generality, $i_1$ to $i_q$ are adjacent to $c_1$. Then $i_{q+1}$ to $i_{\ell}$ are in the second layer of the tree $T$. Furthermore, suppose $c_2$ to $c_k$ are indexed in the order of occurrence in the BFS order. Note that BFS may choose $i_1$ to $i_q$ in arbitrary order before the last clique vertex is chosen.
	
	Now, we construct an MNS order $\sigma$, such that the \cf-tree of $\sigma$ is $T$. We simply pick $c_1$ to $c_k$ in ascending order, that is, we start with the same root $c_1$, followed by the clique vertices in unchanged order. Since all vertices in the clique have the same neighborhood of visited vertices at every step and none of the $i_x$ has a larger neighborhood, this does not contradict the MNS search paradigm. Finally, we add the independent set vertices to $\sigma$. Here, we have to choose the independent vertices with larger neighborhoods first. As the whole neighborhood of each of these vertices is already chosen, this does not change the edges of the tree, i.e., the first visited neighbor. Since the neighbors of the independent set vertices are visited in the same order as in the BFS, the same \cf-tree $T$ is generated.
	
	Now suppose that  $\tau$ starts with an independent vertex and, without loss of generality, 
	we label the root of the search tree $T$ by $i_1$. Then the neighbors of $i_1$, say $c_1$ to $c_q$ are in the first layer of the search tree. All other clique vertices and all independent set vertices which are neighbors of $c_1$ to $c_q$ are in the second layer of the \cf-tree $T$. Finally, all remaining independent set vertices are in third layer. Again note that $c_1$ to $c_k$ are assumed to be indexed in the order of occurrence in the BFS order.
	
	Again, a similar order $\sigma$, now starting with $i_1$, followed by $c_1$ to $c_k$ in order of the indices, and afterwards followed by $i_2$ to $i_{\ell}$, respecting neighborhood inclusions, yields the same tree $T$ and it is an MNS order analogous to the above argumentation.
	
	The proof for the other direction can be achieved in the same way. The proofs for MCS, LBFS, and LDFS also follow the same pattern.
\end{proof}

As the \cf-tree problem can be solved in linear time for BFS~\cite{manber1990}, this, therefore, also holds for the other searches.

\begin{corollary}
	The \cf-tree problem of MNS, MCS, LBFS and LDFS can be solved in linear time.
\end{corollary}

In order to fully characterize \cl-trees on split graphs for all the investigated MNS-type searches, we first need two lemmas about their search orders. The first is a typical 3-point condition given by Corneil and Krueger~\cite{corneil2008unified}.

\begin{lemma}\cite{corneil2008unified}\label{lemma:mns_3point}
An ordering $\sigma$ of $V$ is an MNS-ordering if and only if the following statement holds: If $a \prec_{\sigma} b \prec_{\sigma} c$ and $ac \in E$ and $ab \notin E$, then there exists a vertex $d$ with $d \prec_{\sigma} b$ and $db \in E$ and $dc \notin E$.
\end{lemma}

The following lemma gives some information about the position of elements of the independent set $I$ in an MNS-ordering of a split graph. We show that, whenever a vertex $v$ of $I$ is to the left of some vertex of the clique $C$, every vertex of $C$ to the left of $v$ has to be a neighbor of $v$ and all remaining neighbors of $v$ have to be chosen directly after $v$.

\begin{lemma}\label{lemma:split_order}
Let $G = (V,E)$ be a split graph with clique $C$ and independent set $I$. Let $\sigma = (v_1, \ldots, v_n)$ be an ordering of $V$. If $\sigma$ is an MNS-ordering, then it holds for every pair of vertices $v_i \in I$ and $v_j \in C$ with $j > i$ that:
\begin{enumerate}
\item $\{v_1, \ldots, v_{i-1}\} \cap C \subseteq N(v_i)$ with $|\{v_1, \ldots, v_{i-1}\} \cap C| = l$
\item $v_{i+1}, \ldots, v_{deg(v_i) - l} \subseteq N(v_i)$
\end{enumerate}
\end{lemma}

\begin{proof}
Assume that $\sigma$ is an MNS-ordering and does not fulfill one of the two conditions, i.e., there is a pair of vertices $v_i \in I$ and $v_j \in C$ with $j > i$ such that one of the conditions is not fulfilled. 

Suppose there is a vertex $v_k \in C$ with $k < i$ and $v_kv_i \notin E$. As $v_k \prec_{\sigma} v_i \prec_{\sigma} v_j$ and $ v_k v_j \in E $ and $ v_k v_i \notin E $, it follows from Lemma~\ref{lemma:mns_3point} that there must be a vertex $d$ such that $dv_i \in E$ but $dv_j \notin E$. Since $v_i \in I$ and $v_j \in C$, such a vertex cannot exist. Hence, we know that the first statement holds. 

Now, assume that the second condition does not hold and let $v_i$ be the $ \sigma $-leftmost vertex for which it fails. Thus, there is a neighbor $c$ of $v_i$ and a non-neighbor $b$ of $v_i$ such that $v_i \prec_{\sigma} b \prec_{\sigma} c$. Let $b$ the first non-neighbor of $v_i$ to the right of $v_i$ on $ P $. Again, we know, due to Lemma~\ref{lemma:mns_3point} and the choice of $b$, that there must be a vertex $d \prec_{\sigma} v_i$ with $db \in E$ and $dc \notin E$. Since $c \in C$, $d$ must be an element of $I$. Then, however, the second statement does not hold for $d$, since between $d$ and its neighbor $b$ the search has taken the non-neighbor $v_i$. This is a contradiction to the minimality of $v_i$. 
%
\end{proof}

\begin{theorem}\label{theorem:split_ltree}
A tree $T$ is an \cl-tree of MNS (MCS, LDFS, LBFS) on a split graph $G = (V, E)$ with clique $C$ and independent set $I$ if and only if:
\begin{enumerate}
\item $T$ is a caterpillar tree consisting of a set of leaves $ L $ and a dominating path $P = (v_1, \ldots, v_k)$ which contains every vertex of $C$.
\item It holds for every leaf $w \in L$ with a neighbor $ v_i $ in $T$ that $ wv_j \notin E(G) $ for $ j>i $.
\item It holds for every $v_i \in I$ that: 
\begin{enumerate}
\item $\{v_1, \ldots, v_{i-1}\} \cap C \subseteq N(v_i)$ with $|\{v_1, \ldots, v_{i-1}\} \cap C| = l$
\item $v_{i+1}, \ldots, v_{deg(v_i) - l} \subseteq N(v_i)$
\end{enumerate}
\end{enumerate}
\end{theorem}

\begin{proof}
First, we show that the three conditions stated in the theorem are necessary. Assume that $T$ is an \cl-tree of one of the searches, but not a caterpillar tree. We assume, that the tree is rooted in the starting vertex of the search. If $ T $ is not a caterpillar, there exists a vertex $v$ that has two children $ u $ and $ w $ in $ T $ which, in turn, also have two children $ u' $ and $ w' $, respectively. We now show that $uw \notin E$: It is clear that we have to take $u$ and $w$ after $v$, as $ T $ is rooted in the starting vertex of the search. Without loss of generality, we assume that $u$ has been visited first. If $u$ is adjacent to $w$, then the edge $vw$ cannot be part of the tree. Thus, we know that at least one of $u$ and $w$ has to be in $I$ and that $ v \in C $. 

We first assume that $u \in I$ and $w \in C$ and, as a result, $ u' \in C $. By Lemma~\ref{lemma:split_order} it follows that $u$ and $u'$ have to be taken before $w$. Since $u'w \in E$, the edge $vw$ cannot be part of the tree.

Let us now assume that $u$ and $w$ are elements of $I$ and $u$ is to the left of $w$. Then $ u', w' \in C $ and $u'$ has to be taken before $w$ by Lemma~\ref{lemma:split_order}. Since $w$ must be to the left of $w'$, the vertex $u'$ has to be a neighbor of $w$. Thus, the edge $vw$ was not inserted into the tree. Therefore, it follows that such a vertex $v$ cannot exist and the tree must be a caterpillar tree.

For the first statement, it remains to show that each $v \in C$ is part of $P$. Assume, $v \in C$ is a leaf and on both sides of the neighbor of $v$ on $ P $ there is a vertex of $ C $. Let $w$ be the neighbor of $v$ in the tree and let $u$ be the right neighbor of $w$ in $P$. Then using the same argumentation as above $vu \notin E$ and thus $u$ is an element of $I$. Furthermore, $u$ must be to the left of $v$ in $\sigma$. Since $u$ has at least one further neighbor in $T$, this neighbor is adjacent to $v$ and to the left of $v$, due to Lemma~\ref{lemma:split_order}. Thus, the edge $vw$ cannot be an element of $T$.

For the second and third conditions we first show that, without loss of generality, the starting vertex of the search can be assumed to be $ v_1 $. 

Suppose, that the search begins in $ v_i \in P $ with $ 1 < i < k $ and that there are two vertices $ v_l, v_j  \in C$ with $ l < i < j $; let these be leftmost, respectively, rightmost with this property. Furthermore, we assume, without loss of generality, that $ v_l $ is visited before $ v_j $. Let $ x $ be the predecessor of  $ v_j $ on $ P $. If $ x = v_i $, we have a contradiction to $ T $ being an \cl-tree, as $ v_l $ was visited after $ v_i $. Therefore, let $ y $ be the predecessor of $ x $. Again, suppose that $ y= v_i $. The vertex $ v_l $ must have been visited before $ x $, as otherwise $ T $ cannot be an \cl-tree. Due to Lemma~\ref{lemma:split_order}, $ v_l $ must be adjacent to $ x $; this is, again, a contradiction to $ T $ being an \cl-tree. As $ I $ is an independent set, either $ x $ or $ y $ must be in $ C $; this is a contradiction to the choice of $ v_j $. As a result, we can assume, without loss of generality, that all vertices of $ C $ are to the right of the starting vertex in $ P $. If the starting vertex is an element of $ I $, then we see that it must be equal to $ v_1 $. If the starting vertex is an element of $ C $, then it is possible that exactly one vertex of $ I $ is to its left. This vertex must be a leaf in $ G $, as $ T $ is an \cl-tree and we can assume that it is in $ L $, without loss of generality.

If, on the other hand, the start vertex $ r $ is not in $ P $, then it must have a neighbor $ v_i \in P \cap C$ which is the second vertex of the search. Due to the above, we see that all other vertices of $ C $ can be assumed to be to the right of $ v_i $ in $ P $, and, therefore, there is a path $ P' $ fulfilling all the necessary conditions beginning in $ r $.

Hence, the second statement follows from the definition of $\cl$-trees and the third statement follows from Lemma~\ref{lemma:split_order}.

It remains to show that the three conditions are also sufficient. Suppose that we are given a tree $ T $, consisting of a path $ P=(v_1, \ldots , v_k) $ and a set of leaves $ L $, which satisfies all three properties; we then construct MNS, MCS, LDFS and LBFS orderings which generate the \cl-tree $ T $. We consider the ordering $\sigma = (v_1, \ldots, v_k, l_1, \ldots, l_r)$ with $l_i \in L$.

First, we show that all vertices of $ P $ can be visited consecutively in that order by all of the investigated searches. If all vertices of $ P $ are elements of the clique, then this is obvious, as all of these searches can visit a clique in the beginning of the search in an arbitrary order. Now suppose that $ \sigma_i=(v_1, \ldots , v_i) $ is a correct search of one of the given types. We show that $ v_{i+1} $ has maximum label at this point of the search.

If $ v_{i+1} $ is an element of $ C $, it is adjacent to all vertices of $ C \cap \{v_1, \ldots , v_i\} $. Furthermore, there cannot be a vertex $ w \in I\cap \{v_1, \ldots , v_i\}  $ that is not adjacent to $ v_{i+1} $, but to some other unvisited vertex, due to condition 3b). This implies that $ v_{i+1} $ has largest label for all these searches.

Suppose that $ v_{i+1} $ is an element of $ I $. Due to condition 3a), we see again that $ v_{i+1} $ is adjacent to all vertices of $ C \cap \{v_1, \ldots , v_i\} $. As a result of condition 3b), we see that there cannot be any unvisited vertex that is adjacent to a vertex from $ I \cap \{v_1, \ldots , v_i\} $. This implies that $ v_{i+1} $ has largest label for all these searches.

As soon as the path $ P $ has been completely visited, the order in which the remaining vertices of $ I $ are chosen does not have any impact on the resulting search tree, as all neighbors of these vertices have already been chosen. Therefore, we can visit these in any arbitrary ordering that adheres to the given search paradigm. Finally, due to condition~2, the tree resulting from this search coincides with $ T $. 
\end{proof}

The three conditions of Theorem~\ref{theorem:split_ltree} can be checked in linear time. As caterpillar trees are recognizable in linear time, it suffices to define the correct path $P$. To this end, we have to decide whether one of the endpoints must be a vertex from the independent set. Due to Theorem~\ref{theorem:split_ltree}, there can be at most one vertex from the independent set at one of the two ends which is not a leaf in $G$ and this vertex must be the start vertex of the search. If such a vertex is a leaf in $ G $, then we can assume that it is not contained in $ P $.

If $P$ does not begin in a vertex of $ I $, conditions~2) and~3) must be checked for both directions of $ P $. It is easy to decide the second condition by simply checking the indices of the neighbors of vertices in $ L $. To check the third condition, we first place the vertices of $C$ in a separate list according to their ordering in $P$. Then, we mark the neighbors of $v$ for every vertex $v \in I \cap P$ and check, whether all vertices of $C$ that appear before $v$ in $P$ are neighbors of $v$. The remaining neighbors of $v$ must follow $v$ in $P$ directly. All of these operations can be done in $\mathcal{O}(deg(v))$, resulting in a combined running time of $ \mathcal{O}(|E|) $.

\begin{corollary}
	The \cl-tree problem of MNS, MCS, LBFS and LDFS can be solved in linear time.
\end{corollary}


\section{Conclusion}

We have shown that the $\cf$-tree problem is \NP-complete for LBFS, MCS and MNS. Furthermore, we have given polynomial time algorithms for the $\cl$-tree problem of LDFS and for both the $\cf$-tree and the $\cl$-tree problems of LBFS, LDFS, MCS and MNS on split graphs. To the best of our knowledge, no hardness results for the $\cl$-tree problem were known before. Thus, the question arises whether the $\cl$-tree recognition problem is easy in general for every graph search.

For the end-vertex problem, there are polynomial algorithms for some chordal graph classes besides split graphs (cf.~\cite{beisegel2018end-vertex,charbit2014influence,corneil2010end}). Can these results be transferred to the tree-recognition problem? Up to now, there is no known combination of graph class and search for which the end-vertex problem is easy but the tree-recognition problem is hard.

Moreover, we have considered the search tree recognition problem for labeled, unrooted trees in this paper. 
As a variant of this problem, one could fix the starting vertex of the search, i.e., the input would be a rooted search tree. As we have already seen in Section~\ref{sec:ldfs}, if we can solve the problem with a fixed start vertex in polynomial time, we can also solve the general problem efficiently by solving it for every vertex as the starting point of the search. Nevertheless, it could be possible that the problem without fixed starting vertex is easier than the problem with fixed start vertex. That is, maybe it is easy to find a search order with arbitrary root, that generates the tree, but it is \NP-hard to find one that uses the given root.

As a second variant, one can also consider the unlabeled problem, i.e., no spanning tree is given, but a tree with a matching number of vertices. Thus, we are looking for a search tree which is isomorphic to the given tree. Obviously, this problem is \NP-hard for \cl-trees of DFS, since it includes the hamiltonian path problem. However, it remains open whether there are searches and graph classes where the unlabeled case is easy or even easier than the labeled one. 

In the literature, spanning trees with special properties and corresponding optimization problems are well studied. Examples are the maximum leaf spanning tree problem~\cite{garey2002computers} and distance approximating spanning trees~\cite{prisner1997distance}. Are there graph classes where search trees of the investigated graph searches solve or at lead to an approximate solution of such problems?

	\bibliographystyle{plain}
	\bibliography{tree-recognition}

\end{document}